\documentclass{article}

\usepackage{amsmath}
\usepackage{amssymb}
\usepackage{amsthm}
\usepackage{xspace}

\theoremstyle{definition}

\newtheorem{definition}{Definition}
\newtheorem{theorem}{Theorem}
\newtheorem{lemma}{Lemma}
\newtheorem{proposition}{Proposition}

\newcommand{\eps}{\varepsilon}

\newcommand{\E}{\mathrm{E}}

\newcommand{\B}{\{0,1\}}

\newcommand{\Bn}{\{0,1\}^n}
\newcommand{\Bm}{\{0,1\}^m}
\newcommand{\rf}[1]{{\ref{#1}}}

\newcommand{\OO}{\tilde{O}}
\newcommand{\h}{\mathsf{h}}

\begin{document}

\title{The Relative Exponential Time Complexity of Approximate Counting Satisfying Assignments}

\author{Patrick Traxler\thanks{This work was partially done during the authors doctoral studies at ETH Zurich \cite{Traxler10} and supported by the Swiss National Science Foundation SNF under project 200021-118001/1.}}


\maketitle

\begin{abstract} 
We study the exponential time complexity of approximate counting satisfying assignments of CNFs. We reduce the problem to deciding satisfiability of a CNF. Our reduction preserves the number of variables of the input formula and thus also preserves the exponential complexity of approximate counting.

Our algorithm is also similar to an algorithm which works particular well in practice for which however no approximation guarantee was known. Towards an analysis of our reduction we provide a new inequality similar to the Bonami-Beckner hypercontractive inequality. 



\end{abstract}

\section{Introduction}

We analyze the approximation ratio of an algorithm for approximately counting solutions of a CNF. The idea of our algorithm goes back to Stockmeyer. Stockmeyer \cite{Stockmeyer85} shows that approximately counting witnesses of any NP-relation is possible in randomized polynomial time given access to a $\Sigma_2$P-oracle. It is known that we only need an NP-oracle if we apply the Left-Over Hashing Lemma of Impagliazzo, Levin, and Luby \cite{ImpagliazzoLL89} which we discuss below. The use of an NP-oracle is necessary, unless P = NP. Stockmeyer's result and its improvement provides us with a first relation between deciding satisfiability and approximately counting solutions, a seemingly harder problem. 

\subsection{Exponential Time Complexity}


The motivation of our results comes from exponential time complexity. Impagliazzio, Paturi, and Zane \cite{ImpagliazzoPZ01} develop a structural approach to classify NP-complete problems according to their exact time complexity. They formulate and prove the Sparsification Lemma for $k$-CNFs. This lemma allows us to use almost all known polynomial time reductions from the theory of NP-completeness to obtain exponential hardness results. There are however problems for which the sparsification lemma and standard NP-reductions do not yield meaningfull results. Relating the exact complexity of approximately counting CNF solutions and the complexity of SAT is such a problem. We show: 

\emph{Let $c>0$ and assume there is an algorithm for SAT with running time $\OO(2^{cn})$. For any $\delta>0$, there is an algorithm which outputs with high probability in time  $\OO(2^{(c+\delta)n})$ the approximation $\tilde{s}$ for the number of solutions $s$ of an input CNF such that $$(1-2^{-\alpha n})\,s\leq \tilde{s}\leq (1+2^{-\alpha n})\,s$$ with $\alpha=\Omega(\frac{\delta^2}{\log(\frac{1}{\delta})})$.} 

It is not clear if this approximation problem is in $BPP^{NP}$ because of the super-polynomially small approximation error. An improvement of the approximation error would yield a similar reduction from $\#$SAT to SAT.

A further application of our algorithm is to sample a solution approximately uniformly from the set of all solutions \cite{JerrumVV86}. The approximation error is again subexponentially small in $n$. The reduction in \cite{JerrumVV86} preserves the number of variables. 

We can get also a result similar to Stockmeyer's result. For any problem in parameterized SNP \cite{ImpagliazzoPZ01} -- an appropriate refinement and subset of NP -- we can define its counting version. Every such problem reduces by our result and the sparsification lemma to SAT at the expense of an increase of $n$ to $O(n)$ variables. Here, $n$ may be the number of vertices in the graph coloring problem or a similar parameter \cite{ImpagliazzoPZ01}. We just have to observe that the sparsification lemma preserves the number of solutions. 



\subsection{A Practical Algorithm}

Stockmeyer's idea was implemented in \cite{GomesSS06}. Gomes et al.\ \cite{GomesSS06} provide an implementation of a reduction which uses a SAT-solver to answer oracle queries. The algorithm of Gomes et al.\ \cite{GomesSS06} is almost the same as our algorithm. It preserves the number of variables  and the maximum clause width is small. These properties seem to be crucial for a fast implementation, in particular, for the SAT-solver to work fast. 





Gomes et al.\ \cite{GomesSS06} compare empirically the running time of their algorithm to the running time of exact counting algorithms. Their algorithm performs well on the tested hard instances and actually outperforms exact counting algorithms. The output values seem to be good approximations. The reason for this is not understood by theoretical means yet. A bound on the approximation ratio is not known.

Because there are only small differences between our algorithm and the algorithm of Gomes et al.\ \cite{GomesSS06}, our bound on the approximation guarantee may be considered as a theoretical justification for the quality of the algorithm of Gomes et al.\ \cite{GomesSS06}. We do not attempt here to explain why the SAT-solver is able to handle the generated instances well.

Another algorithm for the $k$-CNF case with theoretical bounds was proposed by Thurely \cite{Thurely12}.



\subsection{Comparison to the Left-Over Hashing Lemma}

A possible reduction from approximate counting to satisfiability testing works roughly as follows. We assume to have a procedure which takes as input a CNF $F$ with $n$ variables and a parameter $m$. It outputs a CNF $F\wedge G_m$ such that the number of solutions of $F\wedge G_m$ times $2^m$ is approximately the number of solutions of $F$. We apply this procedure for $m=1,...,n$ and stop as soon as $F\wedge G_m$ is unsatisfiable. Using the information when the algorithm stops we can get a good approximation. 

The construction of $G_m$ reduces to the following randomness extraction problem. We have given a random point $x\in\Bn$ and want a function $h:\Bn\rightarrow\Bm$ such that $h(x)$ is almost uniform. We think of $h$ as $m$ functions $(h_1,...,h_m)$ and additionally require that each $h_i$ depends only on few coordinates. We use the later property to efficiently encode $h$ as a CNF in such a way that the encoding and the input CNF $F$ have the same number of variables. Stockmeyer's result and its improvement can not be adapted easily to get such an efficient encoding. The crucial difference of our approach to the original approach are the bounds on the locality of the hash function. Our analysis is Fourier-analytic whereas the proof of Left-Over Hashing Lemma \cite{ImpagliazzoLL89} uses probabilistic techniques. 





Impagliazzo et al.\ \cite{ImpagliazzoLL89} show that any pairwise independent\footnote{Pairwise independence means here that $\Pr_{h\sim\mathcal{H}_\text{ind}}(h(x_1)=y_1,h(x_2)=y_2) = 2^{-2m}$ for any $x_1,x_2\in\Bn$, $x_1\not=x_2$, and $y_1,y_2\in\Bm$. A Bernoulli matrix with bias $\frac{1}{2}$ induces a for example a pairwise independent family.} family $\mathcal{H}_\text{ind}$ of functions of the form $\Bn\rightarrow\Bm$ satisfies the following \emph{extraction property}: Fix a distribution $f$ over the cube $\Bn$ with bounded min-entropy\footnote{See Sec.\ \rf{section:preliminaries}.} $\Omega(m+\log(1/\eps))$ and $y\in\Bm$. Then,  
\begin{align*}
	\Pr_{h\sim\mathcal{H}_\text{ind}}(|\Pr_{x\sim f}(h(x)=y) - 2^{-m}| \leq \eps\,2^{-m})\geq 0.1.
\end{align*}
This result, in a slightly more general form \cite{ImpagliazzoLL89}, is called the \emph{Left-Over Hashing Lemma}. We want for our applications that $h$, seen as a random function, has besides the extraction property a couple of additional properties. The most important being that $h_i$ is a Boolean function depending on at most $k$ coordinates. This is what we call a \emph{local hash function}. These hash functions are however not necessarily pairwise independent. This leads to a substantial problem. The proof of the Left-Over Hashing Lemma relies on pairwise independence since it allows an application of Chebyshev's Inequality. In its proof we define the random variable $X=X(h):=\Pr_{x\sim f}(h(x)=y)$. Its expected value is $2^{-m}$. This still holds in our situation. Its variance can be however too large for an application of Chebyshev's Inequality. To circumvent the use of Chebyshev's Inequality we formulate the problem in terms of Fourier analysis of Boolean functions. We make use of a close connection between linear hash functions attaining the extraction property and the Fourier spectrum of probability distributions over the cube $\Bn$. 



\subsection{Further related work} Calabro et al.\ \cite{CalabroIKP08} give a probabilistic construction of a "local hash function" without the extraction property. They obtain a similar reduction as the Valiant-Vazirani reduction \cite{ValiantV86}. The extraction property is not necessary for this purpose. Gavinsky et al.\ \cite{GavinskyKKRW08} obtain a local hash function via the Bonami-Beckner Hypercontractive Inequality. However only for $|A|\geq 2^{n-O(\sqrt{n})}$. We remark that the motivations and applications in \cite{GavinskyKKRW08} are different from ours. 


We lend the term extraction property from Goldreich \& Wigderson \cite{GoldreichW97}. The goal in \cite{GoldreichW97} is to find small families of hash functions to reduce the amount of random bits needed to sample the hash function. In a more restrictive setting motivated by problems in cryptography also locality plays an important role. Vadhan \cite{Vadhan04} studies locally computable extractors. A locally computable extractor is essentially the same as a local hash function but with the difference that the functions $h_1,...,h_m$ which constitute the hash function may depend in total on $O(m)$ coordinates. A notion of locality (for pseudorandom generators) which is closer to ours is studied in the context of cryptography \cite{ApplebaumIK06} and inapproximability \cite{ApplebaumIK06-2}.

The Bonami-Beckner Hypercontractive Inequality, credited to Bonami \cite{Bonami70} and Beckner \cite{Beckner75}, found several diverse applications.  See \cite{deWolf08,ODonnell08} for further references.

\section{Preliminaries}
\label{section:preliminaries}

We make the following conventions. We assume uniform sampling if we sample from a set without specifying the distribution. We also use a special $O(\cdot)$ notation for estimating the running time of algorithms. We suppress a polynomial factor depending on the input size by writing $\OO(\cdot)$. As an example, SAT can be solved in time $\OO(2^n)$. We denote the logarithm with base $2$ by $\log(\cdot)$ and the logarithm naturalis by $\ln(\cdot)$.
 

A $\kappa$-\emph{junta} is a Boolean function which depends on at most $\kappa$ out of $n$ coordinates. We extend this notion to functions $h:\Bn\rightarrow\Bm$, $h=(h_1,...,h_m)$, by requiring that $h_i$ is a $\kappa$-junta for every $i\in[m]$. A Boolean function $f:\Bn\rightarrow\mathbb{R}$ is a \emph{distribution} iff all values of $f$ are non-negative and sum up to $1$. It has \emph{min-entropy} $t$ iff $t$ is the largest $r$ with $f(x)\leq 2^{-r}$ for all $x\in\Bn$. The \emph{relative min-entropy} $\tilde{t}$ is defined as $\tilde{t}:=t/n$. A distribution $f$ is $t$-\emph{flat} iff $f(x)=2^{-t}$ or $f(x)=0$ for all $x\in\Bn$.  


\begin{definition}
Let $0<p_1,p_2\leq 1$. Let $\mathcal{D}$ be a distribution over functions of the form $\Bn\rightarrow\Bm$. A random function $h$ is called $\kappa$-\emph{local} with probability $p_1$ iff 
\begin{align*}
	\Pr_{h\sim\mathcal{D}}(h\text{ is $\kappa$-local}) \geq p_1.
\end{align*}
It is called a $(t_0,\eps)$-\emph{hash function} (\emph{for flat distributions}) with probability $p_2$ iff
\begin{align*}
	\Pr_{h\sim\mathcal{D}}(|\Pr_{x\sim f}(h(x)=y) - 2^{-m}| \leq \eps\,2^{-m}) \geq p_2
\end{align*}
for every $y\in\Bm$ and every (flat) distribution $f$ of min-entropy $t$ with $t_0\leq t\leq n$.
\end{definition}


\section{Local Hash functions: Construction and Analysis}
\label{section:construction}


We start with the definition/construction of the two hash functions $\h$ and $\h^c$. After this we discuss a basic connection between Fourier coefficients of distributions and the special case of linear hash functions with a one-dimensional range. We generalize this finally to functions with the high-dimensional range $\B^m$.\\

\noindent\textbf{Construction of $\h$:} For $i=1,...,m$: Choose a set $S_i\sim \mu_p$. Define $\h_i(x):=\bigoplus_{j\in S_i}x_j$. The hash function is $\h:=(\h_1,...,\h_m)$.\\

\noindent In other words, $\h$ is the linear map given by a Bernoulli matrix with bias $p$.\\

\noindent\textbf{Construction of $\h^c$:} Fix $k$. For $i=1,...,m$: Choose a set $S_i\sim \{S:S\subseteq [n],\, |S|=k\}$. Define $\h^c_i(x):=\bigoplus_{j\in S_i} x_j$. The hash function is $\h^c:=(\h^c_1,...,\h^c_m)$.\\



\subsection{Hashing, Randomness Extraction, and the discrete Fourier transform} 
\label{subsection:transform}

We start with recalling basics from Fourier analysis of Boolean functions. The \emph{Fourier transform} of Boolean functions is a functional which maps $f:\Bn\rightarrow\mathbb{R}$ to $\widehat{f}:2^{[n]}\rightarrow\mathbb{R}$ and which we define by $\widehat{f}(S):= \E_{x\sim\Bn}(f(x)\,(-1)^{\bigoplus_{i\in S}x_i})$, $S\subseteq[n]$. We will study the following \emph{normalized Fourier transform} given by $\widetilde{f}(S) := 2^{n-1}\,\widehat{f}(S)$. We call the values of $\widehat{f}$ \emph{Fourier coefficients} and the collection of Fourier coefficients the \emph{Fourier spectrum} of $f$. 

We can rewrite normalized Fourier coefficients to see the connection to hashing and randomness extraction. We define $\bigoplus_{i\in \{\}}x_i:=0$.

\begin{lemma}
	\label{lemma:biased}
	Let $f:\Bn\rightarrow\mathbb{R}$ be a distribution. For any $S\subseteq[n]$,
	\begin{align*}
		\widetilde{f}(S) &= \Pr_{x\sim f}(\bigoplus_{i\in S}x_i=0)-\frac{1}{2}=\frac{1}{2} - \Pr_{x\sim f}(\bigoplus_{i\in S}x_i=1).
	\end{align*}
\end{lemma}


We may think of $\bigoplus_{i\in S}x_i$ as a single bit which we extract from $f$. We are interested in how close to a uniformly distributed bit it is. There is also a combinatorial interpretation of randomness extraction which we are going to use subsequently. We define for non-empty $A\subseteq\Bn$ the flat distribution $f_A(x):=\frac{1}{|A|}$ if $x\in A$ and $0$ otherwise. We want a random hash function $h:\Bn\rightarrow\B$ such that for every not too small $A\subseteq\Bn$ and $b\in\B$, $\Pr_{h}\left(\left|\Pr_{x\sim f_A}(h(x)=b)-\frac{1}{2}\right|\text{ is small}\right)$ is large. This is the same as saying that the probability of the event $|A\cap\{x\in A:h(x)=b\}|\approx\frac{|A|}{2}$ should be large. In words, the hyperplane in $\mathbb{F}_2^n$ induced by $h$ separates $A$ in roughly equal sized parts. 


\subsection{Analysis of Local Hash Function}

In this section we describe our technical tools for analyzing linear local hash functions. We show how to apply them on the example of the two functions $\h$ and $\h^c$. The first result we need is an inequality similar to the hypercontractive inequality for Boolean functions. We prove actually a more general inequality. It allows us to analyze linear and local hash functions with a one-dimensional range. For the generalization to functions with a high-dimensional range we use a different technique.  


\subsubsection{An Inequality}


We give an outline of the proof. The \emph{support} of a function $g:\Bn\rightarrow\mathbb{R}$ is the set of all points with a non-zero value and denoted by $\mathrm{Supp}(g)$. The norms below are w.r.t.\ the counting measure. Define
\begin{equation*}
 A(\alpha,p):=\sup_{0\leq x\leq 1} \frac{\|(1-2\,p\,x,1-2\,p\,(1-x))\|_\frac{1}{\alpha p}}{\|(x,1-x)\|_{_\frac{1}{1-\alpha p}}}.
\end{equation*}
\begin{lemma}
\label{lemma:contractive}
Let $f,g:\Bn\rightarrow\{-1,0,1\}$, $0< p\leq \frac{1}{2}$, and $0<\alpha\leq 1$. Let $\tilde{A}(\alpha,p)$ be such that $\max(A(\alpha,p),(1-p)\,4^{\alpha p})\leq\tilde{A}(\alpha,p)$. Then,
\begin{equation*}
 \E_{S\sim \mu_p}(\hat{f}(S)\,\hat{g}(S))\leq 4^{-n}\,\tilde{A}(\alpha,p)^n\,(|\mathrm{Supp}(f)|\cdot |\mathrm{Supp}(g)|)^{1-\alpha p}.
\end{equation*}
\end{lemma}

The previous lemma is shown by induction over $n$. In its proof we work explicitly with the Bernoulli distribution $S$ is chosen from and avoid entirely the use of the (noise) operator as in \cite{Beckner75}. The purpose is to decompose in the induction step the $n$-dimensional functions $f$ and $g$ into $(n-1)$-dimensional functions with the same range $\{-1,0,1\}$. Preserving the range seems to be an interesting benefit of our new proof.

The following estimation is the reason why it makes sense to introduce the new quantity $\alpha$ which does not occur in \cite{Beckner75}. Setting for example $\alpha = 1/\log(n)$ will make $\tilde{A}(\alpha, p)$ already reasonable small.

\begin{lemma}
\label{lemma:contractive2}
It holds that $A(\alpha,p)\leq \big(1+\,2^{-1/\alpha+8}\big)^{\alpha p}$ for $0<\alpha\leq\frac{1}{9}$, $0<p\leq\frac{1}{2}$.
\end{lemma}

Finally, we arrive at the result we need. Its an application of the previous results together with a result of Chor \& Goldreich \cite{ChorG89}. It seems that the Bonami-Beckner Inequality is too weak for proving it.

\begin{lemma}
\label{lemma:contractive3}
Let $f:\Bn\rightarrow\mathbb{R}$ be a distribution of relative min-entropy $\tilde{t}$, $2^{\tilde{t}n}\in\{1,...,2^n\}$, and $0<p\leq\frac{1}{2}$. Then, $$\E_{S\sim\mu_p}(|\widetilde{f}(S)|)\leq \frac{1}{2}\sqrt{2}^{-p\cdot n\cdot \tilde{t}/\log(512/\tilde{t})}.$$
\end{lemma}

Applying the Bonami-Beckner hypercontractive inequality we get

\begin{lemma}
	\label{lemma:contractive4}
	Let $f:\Bn\rightarrow\mathbb{R}$ be a distribution of min-entropy $t$ with $2^{t}\in\{1,...,2^n\}$, $k$ be a positive integer, and $0<\zeta<1$. Then, 
	\begin{align*}
		\E_{S\sim {[n]\choose k}}(|\widetilde{f}(S)|)\leq \frac{1}{2}\, n^{-(1-\zeta)k/2}\, 2^{(n-t)\,k\,n^{-\zeta}}.
	\end{align*}
\end{lemma}


\subsubsection{High-Dimensional Range}

Our technique for analyzing hash functions of the form $\Bn\rightarrow\B^m$ works as follows. Assume $f$ has min-entropy $t$. Conditioning on an event $E\subseteq\Bn$ yields a new distribution $f'$ with min-entropy $t'$. We can not say much about the relation of $t$ and $t'$ in general. If $E$ is however a hyperplane (in the vector space $\mathbb{F}^n_2$) induced by $\bigoplus_{i\in S}x_i$ then our inequality from above tells us that $t'\approx t-1$ in the expectation, $S\sim\mu_p$. Iterating this step and keeping control of the entropy decay we get our result. This process works as long as we reach some threshold $t_0$ which is essentially determined by the bias $p$.

Formally, the proof is an induction over $m$ and the induction step an application of Lemma \rf{lemma:contractive3}. We apply it to distributions $f_{i}$ which we define inductively for concrete $h^*:\Bn\rightarrow\Bm$. For $i=0$, $f_0:=f$. For $i>0$, $f_i$ is $f_{i-1}$ conditioned on the event $\{x\in\Bn:h^*_i(x)=y_i\}$, i.e., $f_{i}(z):= {\Pr_{x\sim f_{i-1}}(x = z\;|\;h^*_i(x)=y_i)}$. The function $f_i$ is not well defined for every $h^*$ since $\Pr_{x\sim f_{i-1}}(h^*_i(x)=y_i) = 0$ is possible. If this is the case we define $f_j$ to be $0$ on all points and for all $j\geq i$. The following condition excludes this case if $\eta<1$.
\begin{align}
	\label{property:1}
	\forall 1\leq i\leq m:\;|\Pr_{x\sim f_{i-1}}(h^*_i(x)=y_i) - {1}/{2}| \leq {\eta}/{2}
\end{align}

The next lemma allows us to bound the error of approximation, in particular, how far $\Pr_{x\sim f}(h^*(x)=y)$ is from the optimal value $2^{-m}$.

\begin{lemma} 
	\label{lemma:1}
	Let $0<\eta<1$. If Cond.\ \rf{property:1} holds for $h^*$, then\\

	\noindent 1. $(1-\eta)^j\, 2^{-j}\leq \Pr_{x\sim f}(h^*_1(x)=y_1,...,h^*_j(x)=y_j) \leq (1+\eta)^j\, 2^{-j}$, $j=1,...,m$,\\
	
	\noindent 2. $|\Pr_{x\sim f}(h^*(x)=y) - 2^{-m}|\leq 2^{-m}\, ((1+\eta)^m-1)$.
\end{lemma}

In the proof of our main lemma we establish the desired extraction property for $\h$ and $\h^c$.

\begin{lemma}[Main Lemma] \label{lemma:main}
Let $0<\eps<1$, $0<p\leq\frac{1}{2}$. Define $$P(\tilde{t}):=\frac{m}{\eps}\sqrt{2}^{-p n \tilde{t} / \log(512/\tilde{t})}.$$

\noindent\textup{\textbf{Hash Function $\h$}.} If there exists $\tilde{t}_0$ such that $P=P(\tilde{t}_0)<1$ and $\tilde{t}_0 n+m+1\leq n$, then $\h$ is a $(\tilde{t}_0 n+m+1,\eps)$-hash function for flat distributions with probability at least $(1-P)^m>0$.\\ 

\noindent Let $0<\eps<1$, $0<\zeta<1$. Let $k$ be a positive integer. Define $$Q(t):= \frac{m}{\eps}\, n^{-(1-\zeta)k/2}\, 2^{(n-t)\,k\,n^{-\zeta}}.$$
\noindent\textup{\textbf{Hash Function $\h^c$}.}  If there exists $t_0$ such that $Q=Q(t_0)<1$ and $t_0 +m+1\leq n$, then $\h^c$ is a $(t_0+m+1,\eps)$-hash function for flat distributions with probability at least $(1-Q)^m>0$.
\end{lemma}


We argue next that the restriction $p=\Omega(\frac{\log(m)}{n})$ and that the trade-off between the entropy of the distribution and the bias $p$ are essentially optimal. In other words, we can only expect small improvements of the Main Lemma. 

\subsubsection{Rank of Bernoulli Matrices} We recall the combinatorial idea behind hashing. Let $M$ be a Bernoulli matrix with bias $p$ and let $y\in\Bm$. The preimage of $M$, $y$ intersects any large enough subset $A\subseteq\Bn$ in approximately $|A|\cdot 2^{-m}$ points. Let us assume $m=n$. If especially $A=\Bn$ we expect that the linear system $Mx=y$ has one solution in $\mathbb{F}^n_2$. This is is the case iff $M$ has full rank. The threshold for this property is around $\Theta(\frac{\log(n)}{n})$ \cite{Cooper00a}. In particular, the probability that $M$ has full rank can get very small and in which case $M$ fails to have the extraction property with high probability. With respect to this consideration it is not surprising that our probabilistic construction becomes efficient only if $p=\Omega({\frac{\log(n)}{n}})$.

\subsubsection{The Isolation Problem} We will argue next that also the trade-off between the size of $A$, i.e.\ the min-entropy of the corresponding flat distribution, and $p$ is close to optimal. Actually, we can restrict $A$ to be the solution set of a $k$-CNF. The following result is due to Calabro et al. \cite{CalabroIKP08}: For any distribution $\mathcal{D}$ of $k$-CNFs over $n$ variables, there is a satisfiable $k$-CNF $F$ such that $\Pr_{F'\sim\mathcal{D}}(|\mathrm{sol}(F)\cap\mathrm{sol}(F')|=1)\leq 2^{-\Omega(n/k)}$, where $\mathrm{sol}(F)$ ($\mathrm{sol}(F')$) refers to the set of solutions of $F$ ($F'$). The corresponding problem of computing $F'$ is the Isolation Problem for $k$-CNFs \cite{CalabroIKP08}. We show how the Main Lemma relates to a solution of this problem. Let $G$ be a $k$-CNF and let $p=\frac{k}{n}$, $k=\Theta(\kappa \log(\kappa) \log(n))$. The Main Lemma guarantees just that $|\mathrm{sol}(G)\cap\mathrm{sol}(G')|$, $G'$ the CNF-encoding of $\h$, is with high probability within a small interval around $v=2^{O(n/\kappa)}$. We need to define an appropriate distribution $\mathcal{D}_0$ to apply the mentioned result. Chernoff's Inequality guarantees that $\h$ is encodable as a $k$-CNF $G''$ with high probability. We extend $G''$ by constraints (literals) which encode $x_i=0$ or $x_i=1$ as follows. Uniformly at random select a set of $\log(v)$ variables. Uniformly at random set the value of these variables. This defines our distribution $\mathcal{D}_0$. With probability at least $2^{-O(n\cdot\log(\kappa)/\kappa)}$ we get a $O(k)$-CNF $G'$ such that $|\mathrm{sol}(G)\cap\mathrm{sol}(G')|=1$. The reason for this is the following simple to prove fact (Exercise 12.2, pg.\ 152 in \cite{Jukna01}): Let $B\subseteq\Bn$ be non-empty. There exists a set of coordinates $I\subseteq[n]$ and $b\in\B^I$ such that $|I|\leq\log(|B|)$ and $|\{x\in B:x_i=b_i\,\forall i\in I\}|=1$. Note that the construction of $\mathcal{D}_0$ depends only on the parameters $n$, $k$, and $m$, but not on the input $k$-CNF $G$. We can thus apply the result of Calabro et al.\ \cite{CalabroIKP08}. Comparing the lower and and upper bound we see that we are off by a factor $O(\log(k)^2 \log(n))$ in the exponent.

\section{Complexity of Approximate Counting}

The algorithm is depicted in Fig.\ \ref{figure:1}. It is similar to the algorithm of Gomes et al.\ \cite{GomesSS06}. One difference is the construction of $\h$ which is a Bernoulli matrix with bias $p$ in our case. Gomes et al.\ \cite{GomesSS06} select uniformly at random a linear function which depends on exactly $k$ coordinates for every row. Another difference is the output. We output an approximation for the number of solutions. The algorithm of Gomes et al.\ \cite{GomesSS06} outputs a lower and an upper bound. Besides the experimental results, they can show that with high probability the output lower bound is indeed smaller than the number of solutions. They give however no estimation for the quality of the output bounds which would be necessary for bounding the approximation ratio. 

\begin{figure}
\small
Input: CNF $F$ over $n$ variables and a parameter $k$.\vspace{5pt}\\
1. \noindent Set $p:=\frac{k+1}{2 n}$.\\
2. For $l=1,...,n+1$:\\
3. \hspace*{20pt} Repeat $8\lceil\log(n)\rceil$ times:\\ 
4. \hspace*{40pt} Construct $\h$. Select $b\sim\B^l$.\\
5. \hspace*{40pt} If $|S_i|>k$ for some $i\in[l]$ then stop.\\
6. \hspace*{40pt} Let $G$ be the $k$-CNF encoding of $h(x)=b$.\\ 
7. \hspace*{40pt} Record if $F\wedge G$ is satisfiable.\\
8. \hspace*{20pt} If unsatisfiability was recorded more than $4\lceil\log(n)\rceil$ times\\
9. \hspace*{40pt} then output $2^{l-1}$ and stop.\\
10. Output $0$.
\normalfont
\caption{Algorithm \textsc{acount} with access to a SAT-oracle}
\label{figure:1}
\end{figure}

We define algorithm \textsc{acount-constant} similar to \textsc{acount} but with the only difference that it constructs $\h^c$. We stress the fact that our algorithms are easy to implement and that we can amplify the success probability further by repeating the inner loop appropriately. 

\begin{theorem}
 \label{theorem:algorithm}
1. \textbf{(Complexity of Approximate Counting)} Let $c>0$ and assume there is an algorithm for SAT with running time $\OO(2^{cn})$. For any $\delta>0$, there is an algorithm which outputs with high probability in time  $\OO(2^{(c+\delta)n})$ the approximation $\tilde{s}$ for the number of solutions $s$ of an input CNF such that $$(1-2^{-\alpha n})\,s\leq \tilde{s}\leq (1+2^{-\alpha n})\,s$$ with $\alpha=\Omega(\frac{\delta^2}{\log(\frac{1}{\delta})})$.\\ 

2. \textbf{(Algorithm Analysis)} Let $k$ be such that $4\log(16 n) \leq k+1\leq n$ and let $\kappa$ be such that $k+1=\kappa\,\log(512\kappa)\,4\log(16 n)$. Let $s$ be the number of solutions of $F$. The probability that algorithm \textsc{acount} outputs in time $O(n\cdot\log(n)\cdot (n^2 + 2^k\cdot k\cdot n + \mathrm{size}(F)))$ the approximation $\tilde{s}$ such that $$\frac{1}{4}\,2^{-n / \kappa }\,s\leq \tilde{s}\leq 4\,s$$ is at least $1/4$.\\

For constant $k\geq 5$, the probability that algorithm \textsc{acount-constant} outputs the approximation $\tilde{s}$ such that $$\frac{1}{4}\,2^{-n + \frac{\log(n)}{k}\,n^{1-4/k} }\,s\leq \tilde{s}\leq 4\,s$$ is at least $1/4$.
\end{theorem}

\bibliographystyle{plain}
\bibliography{refs}

\newpage

\appendix

\begin{center}
   \huge{\textbf{Appendix}}
\end{center}

\section{Proof of Lemma \rf{lemma:biased}}

\begin{proof}
Let $I$ be the image of $f$ and $X_p:=\{x\in\Bn:f(x)=p\}$, $p\in I$. Let $E:=\{x\in\Bn:\bigoplus_{i\in S}x_i=1\}$.
\begin{align*} 
	&\widetilde{f}(S) =2^{n-1}\, \E_x(f(x)\,(-1)^{\bigoplus_{i\in S}x_i}) =\\
    &=2^{n-1}\sum_{p\in I}p\,\bigg(\Pr_{x}(f(x)=p,\bigoplus_{i\in S}x_i=0)-\Pr_{x}(f(x)=p,\bigoplus_{i\in S}x_i=1)\bigg)=\\
	&=2^{n-1}\sum_{p\in I}p\,\bigg(\Pr_{x\sim\Bn}(f(x)=p)-2\cdot\Pr_{x\sim\Bn}(f(x)=p,x\in E)\bigg)=\\
	&=\sum_{p\in I}p\,\bigg(\frac{1}{2}\,|X_p|-|X_p\cap E|\bigg) = \frac{1}{2}\sum_{p\in I,\,x\in X_p}p - \sum_{p\in I} \sum_{x\in X_p\cap E} p =\\
	&= \frac{1}{2} - \Pr_{x\sim f}(\bigoplus_{i\in S}x_i=1)=\Pr_{x\sim f}(\bigoplus_{i\in S}x_i=0) - \frac{1}{2}.
\end{align*}%
\end{proof}

\section{Proof of Lemma \rf{lemma:contractive}}

\begin{proof}
The proof is by induction on $n$. Let $n=1$. If $f$ or $g$ is the constant $0$ function then the claim holds. There are $8$ remaining functions of the form $\B\rightarrow\{-1,0,1\}$. We start with functions with range $\B$. Let $h_1$ be the identity function, $h_2$ be the function which maps $0$ to $1$, $1$ to $0$, and let $h_3$ be the constant $1$ function. Their Fourier coefficients in order $(\hat{h}_i(\{\}),\hat{h}_i(\{1\}))$ are $(\frac{1}{2},-\frac{1}{2})$, $(\frac{1}{2},\frac{1}{2})$, and $(1,0)$. Avoiding symmetric cases we have $6$ combinations to check. We start our case analysis with $f=g=h_3$:
\begin{equation*}
 1-p\leq 4^{-1}\,\tilde{A}(\alpha,p)\,4^{1-\alpha p} = \tilde{A}(\alpha,p)\,4^{-\alpha p}. 
\end{equation*}
This inequality holds by definition of $\tilde{A}(\alpha,p)$. For the cases $f=h_1,g=h_3$ and $f=h_2,g=h_3$ we have $(1-p)\frac{1}{2}$ on the left-hand side of the inequality:
\begin{equation*}
 (1-p)\,\frac{1}{2}\leq 4^{-1}\,2^{1-\alpha p} = \frac{1}{2}\,2^{-\alpha p}\leq \tilde{A}(\alpha,p)\,\frac{1}{2}\,2^{-\alpha p} 
\end{equation*}
since $\tilde{A}(\alpha,p)\geq 1$. The cases $f=g=h_1$ and $f=g=h_2$ are immediate since the left-hand sides are at most $\frac{1}{4}$. Let $h_4$ be the function which maps $0$ to $-1$ and $1$ to $1$. Its Fourier coefficients are $(\hat{h}_4(\{\}),\hat{h}_4(\{1\}))=(0,-1)$. The claim is thus clearly true for $f=g=h_4$. We reduce the remaining cases to the previous ones by using the linearity of the Fourier transform (multiplying with $-1$). 

Assume that the induction hypothesis holds for $n-1$. For $h:\Bn\rightarrow\{-1,0,1\}$, let $h'_b(x)$ be $1$ if $h(x)=1$ and $x_n=b$, $b\in\B$, and $0$ otherwise. Let $h_b$ be the restriction of $h'_b$ to first $n-1$ coordinates. Let $T\subseteq[n-1]$. It holds that
\begin{align}
 \label{fact:1}
 \hat{h'}_b(T)=& (-1)^b\,\hat{h'}_b(T\cup\{n\})\\
 \label{fact:2}
 \hat{h}_b(T)=&\frac{\hat{h}'_b(T)}{2}.
\end{align}
In what follows, $S$ is chosen from $[n]$ according to $\mu_p$ and $S'$ is chosen from $[n-1]$ also according to $\mu_p$.
\begin{align*}
 &\E_{S}(\hat{f}(S)\,\hat{g}(S)) =\\ 
 &p\,\E_{S}(\hat{f}(S\cup\{n\})\,\hat{g}(S\cup\{n\})\,|n\in S\,) + (1-p)\,\E_{S}(\hat{f}(S)\,\hat{g}(S)\,|n\not\in S\,).
\end{align*}
By the linearity of the Fourier transform (in particular, $\widehat{h_0+h_1}=\hat{h}_0+\hat{h}_1$) and by Eq.\ \ref{fact:1} and \ref{fact:2}, 
\begin{align*}
 &\E_{S}(\hat{f}\hat{g})=\frac{1}{4}\,\big(\E_{S'}(\hat{f_0}\hat{g_0}) + \E_{S'}(\hat{f_1}\hat{g_1}) + (1-2p)(\E_{S'}(\hat{f_0}\hat{g_1})+\E_{S'}(\hat{f_1}\hat{g_0}))\big).
\end{align*}
Define $x_b:=|\mathrm{Supp}(f_b)|$, $y_b:=|\mathrm{Supp}(g_b)|$, $c:=1-\alpha p$, $d_1:=1-2p$, and  $d_2:=\tilde{A}(\alpha,p)$. By the induction hypothesis, 
\begin{align*}
 &\E_{S}(\hat{f}\hat{g})\leq 4^{-n}\,d_2^{n-1}\,( (x_0 y_0)^{c} + (x_1 y_1)^{c} + d_1\,( (x_0 y_1)^{c}+ (x_1 y_0)^{c})).
\end{align*}
We are left with showing that 
\begin{align*}
 &(x_0 y_0)^{c} + (x_1 y_1)^{c} + d_1\,( (x_0 y_1)^{c}+ (x_1 y_0)^{c})\leq d_2\, ((x_0+x_1)(y_0+y_1))^{c}.
\end{align*}
This inequality becomes trivial if at least $2$ variables are $0$ since $d_1\leq 1\leq d_2$. We assume w.l.o.g. that $x_0,x_1,y_0>0$, define $r:=\frac{y_1}{y_0}$, $s:=\frac{x_1}{x_0}$, and divide the inequality by $x_0y_0$. This yields 
\begin{align*}
 1 + d_1\, ( s^c+ r^c) + (r s)^c\leq d_2\, (1+r)^c\,(1+s)^c.
\end{align*}
We define $z(r,s,c):= d_2\, (1+r)^c\,(1+s)^c - 1 - (r s)^c - d_1\, ( s^c+ r^c)$. We are going to show that there exist one $r_0\geq 0$ such that  $\frac{\partial z}{\partial r}(r_0)=0$ first and $\frac{\partial z}{\partial^2 r}(r_0)>0$ for all $s\geq 0$ subsequently. This proves that $r_0$ is a minimum. Finally, we show that $z(r_0,s,p)\geq0$. Differentiating $z$ in $r$ and dividing by $c\,r^{c-1}$ yields
\begin{align*}
 d_2\, (r^{-1}+1)^{c-1}\,(1+s)^{c}  - s^c - d_1=0.
\end{align*}
Resolving for $r$ we get 
\begin{align*}
  r_0 = \bigg(\bigg(\frac{d_1 + s^c}{d_2\,(1+s)^{c}}\bigg)^{1/(c-1)}-1\bigg)^{-1}.
\end{align*}
Define $t:=\frac{d_1 + s^c}{d_2\,(1+s)^{c}}$. Since $s> 0$ and $p\leq\frac{1}{2}$ we conclude that $t> 0$. We also need to show that $t<1$ to conclude that $r_0$ is a positive real. By definition of $d_2$  
\begin{equation*}
  \big(|1-2\,p\,x|^\frac{1}{1-c}+|1-2\,p\,(1-x)|^\frac{1}{1-c}\big)^{1-c} \leq d_2\,\big(|x|^\frac{1}{c}+|1-x|^\frac{1}{c}\big)^{c}.
\end{equation*}
We get this inequality also by multiplying 
\begin{equation}
 \label{eq:3}
 \big((1+d_1\,s^{c})^\frac{1}{1-c}+(d_1+s^{c})^\frac{1}{1-c}\big)^{1-c}\leq d_2\,(1+s)^{c}
\end{equation}
with $(1+s^c)^{-1}$. Then, $0<x=\frac{s^c}{1+s^c}<1$. Note that $d_2$ depends only on $\alpha$ and $p$ and not on $x$. Since $\frac{1+d_1\,s^{c}}{d_2\,(1+s)^{c}}>0$ we conclude that $t<1$.

Next, dividing $\frac{\partial z}{\partial^2 r}$ by $c\,(c-1)\,r^{c-1}$ and noting that $c\,(c-1)<0$ it holds that $\frac{\partial z}{\partial^2 r}(r_0)>0$ iff 
\begin{align*}
  d_2\,(1+r_0^{-1})^{-2+c}\,(1+s)^c-d_1-s^c<0
\end{align*}
iff
\begin{align*}
  (1+r_0^{-1})^{-2+c}=t^{\frac{c-2}{c-1}}<t.
\end{align*}
This inequality holds since $\frac{c-2}{c-1} > 1$ by definition and $0<t<1$ as observed above. We are left with showing that $z(r_0,s,p)\geq0$. It holds that $z(r_0,s,p)\geq0$ iff
\begin{align*}
 (t^{1/(c-1)}-1)^c\,(1+d_1\,s^c)+s^c+d_1\leq d_2\,(1+s)^c\,t^{c/(c-1)}.
\end{align*}
Dividing by $d_2\,(1+s)^c$ yields
\begin{align*}
 (t^{1/(c-1)}-1)^c\,\frac{1+d_1\,s^c}{d_2\,(1+s)^c}\leq t^{c/(c-1)}-\frac{s^c+d_1}{d_2\,(1+s)^c}=t\,(t^{1/(c-1)}-1)
\end{align*}
iff
\begin{align*}
 (t^{1/(c-1)}-1)^{c-1}\leq \frac{d_1+s^c}{1+d_1\,s^c}.
\end{align*}
Dividing by $d_1+s^c$ and rewriting we get Eq.\ \ref{eq:3}.
\end{proof}

\section{Proof of Lemma \rf{lemma:contractive2}}

\begin{lemma} \label{lemma:convexity}
\begin{enumerate}
 \item Let $r\in\mathbb{R}$ and $q\geq 1$. The function $\eta_{r,q}(x)=\|(1-rx,1-r(1-x))\|_q$ is convex in $\mathbb{R}$ and symmetric around $\frac{1}{2}$, i.e., $\eta_{r,q}(\frac{1}{2}-y)=\eta_{r,q}(\frac{1}{2}+y)$. 
 \item $A(\alpha,p)\geq 1$ for every $0<\alpha\leq 1$ and $0< p\leq \frac{1}{2}$.
\end{enumerate}
\end{lemma}

\begin{proof}
We begin with the first claim. Let $0\leq t\leq 1$ and $x,y\in\mathbb{R}$. By Minkoswki's inequality
\begin{align*}
 &t\,\eta(x)+(1-t)\,\eta(y)=\\
 &\| t\,(1-r x, 1- r (1-x)) \|_q+\| (1-t)\,(1-r y, 1-r (1-y)) \|_q\geq\\
 &\| t\,(1-r x, 1-r (1-x)) +(1-t)\,(1-r y, 1-r (1-y)) \|_q=\\
 &\| (1-r (t\,x+(1-t)\,y), 1-r (t\,(1-x)+(1-t)\,(1-y)) \|_q=\\
 &\| (1-r(t\,x+(1-t)\,y), 1-r (1-(t\,x+(1-t)\,y)) \|_q=\eta(t\,x+(1-t)\,y).
\end{align*}
For the second claim, we have to find $x_0$ such that $$\|(1-2\,p\,x_0,1-2\,p\,(1-x_0))\|_\frac{1}{\alpha p}\geq\|(x_0,1-x_0)\|_{_\frac{1}{1-\alpha p}}.$$ Set $x_0=0$. It holds that $\|(1,1-2\,p)\|_\frac{1}{\alpha p}\geq 1 =\|(0,1)\|_{_\frac{1}{1-\alpha p}}$.
\end{proof}

The following proposition is known as Bernoulli's inequality except for the inequality $1+\frac{r x}{2}\leq (1+x)^r$. It can be seen by showing that $(1+x)^r-1-\frac{r x}{2}$ is monotone increasing in $[0,1]$.  
\begin{proposition} \label{proposition:linear}
\begin{enumerate}
 \item If $r\geq 1$ and $x\geq -1$ then $(1+x)^r\geq 1+r x$.
 \item If $0<x,r\leq 1$ then $1+\frac{r x}{2}\leq (1+x)^r\leq 1+r x$. 
\end{enumerate}
\end{proposition}
We will also use the standard estimate $(1-\frac{1}{x})^x\leq \frac{1}{\mathrm{e}}\leq (1-\frac{1}{x+1})^x$, $x\geq 1$, without explicitly mentioning it. 


\begin{proof}
Let $l(x):=\|(1-2\,p\,x,1-2\,p\,(1-x))\|_\frac{1}{\alpha p}$ and $u(x):=\|(x,1-x)\|_{_\frac{1}{1-\alpha p}}$. Both functions are symmetric around $x=\frac{1}{2}$, Lemma \ref{lemma:convexity}. It suffices thus to show the claim for $x\in[0,\frac{1}{2}]$. We simplify the upper bound first. The function $u$ attains its minimum $2^{-\alpha p}$ at $x_0=\frac{1}{2}$, Lemma \ref{lemma:convexity}. Together with Proposition \ref{proposition:linear},
\begin{align*}
 \big(1+\,2^{-1/\alpha+8}\big)^{\alpha p}\,u(x) \geq u(x) + u(x)\,\alpha p\,2^{-1/\alpha+7}\geq u(x) + \alpha p\,2^{-1/\alpha+6}.
\end{align*}
Define $q:=\frac{1}{1-\alpha p}$ and $u_0:=\alpha p\,2^{-1/\alpha+6}$. By Proposition \ref{proposition:linear},
\begin{align*}
 u(x)+u_0\geq x^{q}+(1-x)^{q}+u_0\geq x^{q}+1-q x+u_0=:v(x). 
\end{align*}
The function $v$ is convex and monotone decreasing in $[0,\frac{1}{2}]$ since $\frac{\partial v}{\partial x}=q\,x^{q-1}-q\leq 0$ and $\frac{\partial v}{\partial^2 x}=q(q-1)\,x^{q-2}\geq 0$ for $x\in(0,\frac{1}{2}]$. The idea now is to find a tangent $t$ of $u'$ which lies above $l$. Since $l$ is convex, Lemma \ref{lemma:convexity}, we can show the latter by comparing $l$ and $t$ at $x=0$ and $x=\frac{1}{2}$. The function $v$ has slope $-p$ at $x_0=(1-(1-\alpha p)\,p)^{\frac{1-\alpha p}{\alpha p}}$, 
\begin{equation*}
 \mathrm{exp}\left(-\frac{(1-\alpha p)^2}{\alpha\,(1-(1-\alpha p)p)}\right) \leq x_0\leq \mathrm{exp}\left(-\frac{(1-\alpha p)^2}{\alpha}\right). 
\end{equation*}
We define $t(x):=(v(x_0)+p x_0) - p x$.

Case $l(0)\leq t(0)$: $l(0)=(1+(1-2p)^{\frac{1}{\alpha p}})^{\alpha p}\leq 1+\alpha p\, \mathrm{exp}(-\frac{2}{\alpha})$. On the other side $t(0)=x_0^{q}+1-q x_0+u_0+p x_0\geq x_0^{q}+1-q x_0+u_0\geq x_0^{q}+1 - (1+2\alpha p)\, x_0+u_0$ where we used that $-q\geq -(1+2\alpha p)$. Since $1-\alpha p\geq\sqrt{\frac{1}{\log(\mathrm{e})}}$, $\mathrm{exp}(-\frac{(1-\alpha p)^2}{\alpha})\leq 2^{-\frac{1}{\alpha}}$. It suffices thus to show that 
\begin{equation*}
  0\leq  \mathrm{exp}\left(-\frac{1-\alpha p}{\alpha\,(1-(1-\alpha p)p)}\right)-\mathrm{exp}\left(-\frac{(1-\alpha p)^2}{\alpha}\right)+61\alpha p 2^{-1/\alpha}
\end{equation*}
if 
\begin{equation*}
  0\leq  \mathrm{exp}\left(-\frac{1-\alpha p + 2p}{\alpha}\right)-\mathrm{exp}\left(-\frac{1-2\alpha p}{\alpha}\right)+61\alpha p 2^{-1/\alpha}.
\end{equation*}
We used $1-\alpha p +2 p\geq\frac{1-\alpha p}{1-(1-\alpha p)p}$ here. Multiplying with $\mathrm{exp}(\frac{1}{\alpha}-2p)$ and rearranging yields 
\begin{equation*}
 1-\mathrm{e}^{-p-\frac{2p}{\alpha}}\leq 61\, \mathrm{e}^{-2p}\, \mathrm{e}^{\frac{1}{\alpha}}\, 2^{-\frac{1}{\alpha}}\,\alpha\,p. 
\end{equation*}
Noting that $\frac{1}{3}\leq \mathrm{e}^{-2p}$ and using the estimates $\mathrm{e}^{-p-\frac{2p}{\alpha}}\geq 1-p-\frac{2p}{\alpha}$ and $1+\frac{2}{\alpha}\leq \frac{3}{\alpha}$ we conclude the claim from 
\begin{equation*}
 \frac{3}{\alpha}\leq \frac{61}{3}\, \mathrm{e}^{\frac{1}{\alpha}}\, 2^{-\frac{1}{\alpha}}\,\alpha. 
\end{equation*}

Case $l(\frac{1}{2})\leq t(\frac{1}{2})$: $l(\frac{1}{2})=2^{\alpha p}\,(1-p)\leq \big(\frac{2^\alpha}{\mathrm{e}}\big)^p$. By Proposition \ref{proposition:linear} $\big(1-\frac{p}{2}\big)^{1/p}\geq \frac{1}{2}$ and hence $ \big(\frac{2^\alpha}{\mathrm{e}}\big)^p\leq 1-\frac{p}{2}$. It suffices thus to show $v(x_0)+p x_0\geq 1$. With the the same simplifications as above we get
\begin{equation*}
  0\leq  \mathrm{exp}\left(-\frac{1-\alpha p}{\alpha\,(1-(1-\alpha p)p)}\right)-\mathrm{exp}\left(-\frac{(1-\alpha p)^2}{\alpha}\right)+62\alpha p 2^{-1/\alpha}.
\end{equation*}
 
\end{proof}

\section{Proof of Lemma \rf{lemma:contractive3}}

We need the following fact due to Chor \& Goldreich \cite{ChorG89}. 

\begin{proposition}[Convexity of distributions of bounded min-entropy]
	\label{proposition:convexity}
	Let $t$ be such that $2^t\in\{1,...,2^n\}$. A distribution $f:\Bn\rightarrow\mathbb{R}$ has min-entropy $t$ iff it is a convex combination of $t$-flat distributions $f_1,...,f_L$, i.e., $f = \lambda_1 f_1+...+\lambda_L f_L$ for some positive $\lambda_i$'s with $\lambda_1+...+\lambda_L=1$.
\end{proposition}

\begin{proof} 
Assume $f$ is a $t$-flat distribution. Define $s:=|\{x:f(x)\not=0\}|$ and $g_f:=\lceil f \rceil$, i.e., $f$ rounded up point wise. The range of $g_f$ is $\B$. Applying Lemma \rf{lemma:contractive} and \rf{lemma:contractive2} and using the fact that $\widetilde{f}(S)=\frac{2^n}{2 s}\cdot\widehat{g_f}(S)$, $S\subseteq[n]$,
\begin{align*}
 \E_{S\sim\mu_p} (\widetilde{f}(S)^2) &= (2\,s)^{-2}\,4^{n}\,\E_{S\sim\mu_p} (\widehat{g_f}(S)^2)\leq (2\,s)^{-2}\,(1+2^{-\frac{1}{\alpha}+8})^{\alpha p n}\,s^{2(1-\alpha p)}\leq\\
&\leq 2^{-2}\,\exp\bigg((2^{-\frac{1}{\alpha}+8})\,\alpha p n\bigg)\,s^{-2\alpha p}.
\end{align*}
By Jensen's Inequality,
\begin{align*}
	\E_{S\sim\mu_p} (|\widetilde{f}(S)|) \leq \frac{1}{2}\,\exp\bigg(\frac{1}{2}\,(2^{-\frac{1}{\alpha}+8})\,\alpha p n\bigg)\,s^{-\alpha p}.
\end{align*}
Define $\alpha:=1/\log(512/\tilde{t})$. Note that $s=2^{t}=2^{\tilde{t}n}$ and $\alpha\leq \frac{1}{9}$. Thus,
\begin{align*}
	\E_{S\sim\mu_p} (|\widetilde{f}(S)|) \leq 2^{\alpha p n\tilde{t}\big(\frac{\log(\mathrm{e})}{4}- 1\big)-1}\leq \frac{1}{2} \sqrt{2}^{-\alpha p n \tilde{t}}.
\end{align*}

Let $f$ be a distribution of min-entropy $t$ now. Using the convexity of distributions of bounded min-entropy, Proposition \rf{proposition:convexity}, and the fact that the normalized Fourier transform is a linear functional
\begin{align*}
	\E_{S\sim\mu_p}(|\widetilde{f}(S)|) =\E_{S\sim\mu_p}\left(\left|\sum_{i=1}^L \lambda_i\, \widetilde{f}(S)\right|\right) \leq\sum_{i=1}^L \lambda_i\,\E_{S\sim\mu_p}(|\widetilde{f}(S)|)\leq \frac{1}{2}\sqrt{2}^{-\alpha p \tilde{t} n}.
\end{align*}
\end{proof}

\section{Proof of Lemma \ref{lemma:contractive4}}

\begin{proposition}[Kahn et al.\ \cite{KahnKL88}]
 \label{proposition:kkl}
 Let $f:\Bn\rightarrow\{-1,0,1\}$ and $0\leq\delta\leq 1$. Then,
 \begin{align*}
  \sum_{S\subseteq[n]} \delta^{|S|} \widehat{f}(S)^2 \leq \Pr_{x}(f(x)\not=0)^{\frac{2}{1+\delta}}.
 \end{align*}
\end{proposition}

\begin{proof}
Assume $f$ is a $t$-flat distribution. Define $p:=\Pr_{x}(f(x)\not=0)$. Let $g_f:=\lceil f \rceil$, i.e., $f$ rounded up pointwise. Applying Proposition \ref{proposition:kkl} and using the fact that $\widetilde{f}(S)=\frac{1}{2 p}\cdot\widehat{g_f}(S)$, $S\subseteq[n]$,
\begin{align*}
	&\sum_{S\in {[n]\choose k}} \widetilde{f}(S)^2 = (2\,p)^{-2}\sum_{S\in {[n]\choose k}} \widetilde{g_f}(S)^2\leq  (2\,p)^{-2}\,\delta^{-k}\, p^{\frac{2}{1+\delta}}.
\end{align*}
We recall that for a point $r\in\mathbb{R}^{d}$, $\sum_{i=1}^{d} |r_i|\leq \sqrt{d\;\sum_{i=1}^{d} r_i^2}$. This implies
\begin{align*}
	\sum_{S\in {[n]\choose k}} |\widetilde{f}(S)| \leq \frac{1}{2}\,{n\choose k}^{1/2}\, \delta^{-k/2}\, p^{-\frac{\delta}{1+\delta}}.
\end{align*}
The claim for $t$-flat distributions follows since $p=2^{-(n-t)}$ and since $S$ is chosen uniformly at random from ${[n]\choose k}$. The generalization to distributions of bounded min-entropy follows then from Proposition \ref{proposition:convexity} and the linearity of the Fourier transform. We set $\delta=k/{n^\zeta}$ and use the estimation $|\mathcal{S}_k|\geq (n/k)^k$. Finally, we set $\delta=k/{n^\zeta}$ and use the estimation ${n\choose k}\geq (n/k)^k$.
\end{proof}

\section{Proof of Lemma \rf{lemma:1}}

\begin{proof} 
Define $p_i:=\Pr_{x\sim f_{i-1}}(h^*_i(x)=y_i)$ and $q_i:=\Pr_{x\sim f}(h^*_1(x)=y_1,...,h^*_i(x)=y_i)$. From Cond.\ \rf{property:1}, $(1-\eta)/2\leq p_i \leq (1+\eta)/2$ for $1\leq i\leq m$. In particular, $p_i\not=0$ and $q_i\not=0$ for $1\leq i\leq m$. Thus, 
\begin{align*}
	q_j &= p_{j}\,q_{j-1} = p_{j}\,p_{j-1}\,q_{j-2}	=...= p_{j}\cdot...\cdot p_{1}.
\end{align*}
The first claim follows. Define $q_0:=1$. By the triangle inequality, 
\begin{align*}
	|q_m - 2^{-m}| \leq \sum_{i=1}^m |q_i - \frac{q_{i-1}}{2} |\cdot 2^{-(m-i)}.
\end{align*}
Furthermore, $\frac{1}{q_{i-1}}\cdot |q_i - \frac{q_{i-1}}{2}| =  |p_{i} - \frac{1}{2}|$.
Thus,
\begin{align*}
	|q_m - 2^{-m}| &\leq  \sum_{i=1}^m \frac{|p_{i}-1/2|\cdot q_{i-1}}{2^{m-i}}
	\leq \frac{\eta}{2}\, \sum_{i=1}^m \frac{q_{i-1}}{2^{m-i}}\leq\\
	&\leq \frac{\eta}{2}\, \sum_{i=1}^m \frac{(1+\eta)^{i-1}\, 2^{-{(i-1)}}}{2^{m-i}} = \eta \, \sum_{i=1}^m \frac{(1+\eta)^{i-1}}{2^{m}} .
\end{align*}
Finally, 
\begin{align*}
	|q_m - 2^{-m}| \leq 2^{-m}\,\eta\,\sum_{i=1}^m (1+\eta)^{i-1}=2^{-m} \,((1+\eta)^m-1),
\end{align*}
where we used $\eta\,\sum_{i=1}^m (1+\eta)^{i-1}={(1+\eta)^m-1}$.
\end{proof}

\section{Proof of Lemma \ref{lemma:main}}

\begin{proof}
Let $f$ be a flat distribution of min-entropy $t$ with $t_0=\tilde{t}_0 n\leq t\leq n$. We define $\eta:=\frac{\eps}{2 m}$. We show that $\h$ satisfies Cond.\ \rf{property:1} with probability at least $(1-P)^m$. The induction is over $i=1,...,m$. For $i=1$, we need to show that $|\Pr_{x\sim f}(\h_1(x)=y_1) - {1}/{2}| \leq \eta$ holds with probability at least $1-P$. From Lemma \rf{lemma:biased}, $|\Pr_{x\sim f}(\h^*_1(x)=y_1) - {1}/{2}|=|\widetilde{f}(S^*_1)|$ where $S^*_1$ defines $\h^*_1$. By Markov's Inequality and Lemma \rf{lemma:contractive3}, $\Pr_{S\sim \mu_p}(|\widetilde{f}(S)|\geq\eta)\leq P$. Note that $1-P>0$. 
	
Assume the induction hypothesis holds for $i<m$. We condition on the fact that $(\h_1^*,...,\h_i^*)$ satisfy Cond.\ \rf{property:1}. By Lemma \rf{lemma:1} and observing that flat distributions are closed under conditioning we get that $f_i$ is a flat distribution. We need to show that $|\Pr_{x\sim f_{i}}(\h_{i+1}(x)=y_{i+1}) - {1}/{2}| \leq \eta$ holds with probability at least $1-P$. Again $|\Pr_{x\sim f_{i}}(\h^*_{i+1}(x)=y_{i+1}) - {1}/{2}| = |\widetilde{f}_{i}(S^*_{i+1})|$ where $S^*_{i+1}$ defines $\h^*_{i+1}$. We want to apply Lemma \rf{lemma:contractive3} again. We need to verify that the min-entropy of $f_i$ is not too small. Equivalently, $f_i(z)$ should not be too large for any $z\in\Bn$. By Lemma \rf{lemma:1},
\begin{align*}
	f_{i}(z) &= \Pr_{x\sim f}(x = z\;|\;\h_1(x)=y_1,...,\h_i(x)=y_i)\\
	         &= \Pr_{x\sim f}(x = z,\h_1(x)=y_1,...,\h_i(x)=y_i)\, \Pr_{x\sim f}(\h_1(x)=y_1,...,\h_i(x)=y_i)^{-1}\\
		 &\leq 2^{-t}\; 2^i\; (1-\eta)^{-i}\leq 2^{-t+i+2}.
\end{align*}
The min-entropy of $f_m$ is thus at least $t-i-2\geq t-(m-1)-2\geq (t_0+m+1)-m-1 = t_0$. Applying Lemma \rf{lemma:contractive3} finishes the proof of the claim. 

We showed that $\h$ satisfies Cond.\ \rf{property:1} with probability at least $(1-P)^m$. This implies that $\Pr_\h(|\Pr_{x\sim f}(\h^*(x)=y) - 2^{-m}|\leq \eps\,2^{-m})\geq (1-P)^m$ by Lemma \rf{lemma:1} and since $(1+\eta)^m-1\leq\eps$. This finishes the analysis of $\h$. The analysis for $\h^c$ is the same as for $\h$ but we use Lemma \ref{lemma:contractive4}.
\end{proof}

\section{Proof of Theorem \rf{theorem:algorithm}}

\begin{proof}

\noindent\textbf{Claim 2 (non-constant case).} Let $A$ be the solution set of $F$. Assume $A$ is non-empty and fix $l$. Define $B:=A\cap\{x:\h(x)=b\}$, $\h$ from the $l$-th iteration of \textsc{acount}.\\



\noindent\textbf{Case} $|A| 2^{-l-1} > 2^{n/\kappa}$. Define $f_A(x):=\frac{1}{|A|}$ if $x\in A$ and $0$ otherwise. By the Main Lemma 
\begin{equation*}
 \Pr_{\h,b}\left(\left|\Pr_{x\sim f_A}(\h(x)=b) - 2^{-l}\right|\leq \eps\,2^{-l} \right)\geq \frac{7}{8},
\end{equation*}
i.e., 
\begin{equation}
\label{eq:12}
||B|-|A|\cdot 2^{-l}|\leq |A|\cdot \eps\cdot 2^{-l}.
\end{equation}
We have to calculate $P$ to see this. Set $\eps:=\frac{1}{2}$. First, 
\begin{equation*}
 p=\frac{k+1}{2n}=\frac{2\,\kappa\,\log(512\kappa)\,\log(16 n)}{n}.
\end{equation*}
Thus,
\begin{equation*}
 P = \frac{l}{\eps}\,2^{-\log(16 n^2)\,\kappa\,\log(512\kappa)\,\tilde{t}/\log(512/\tilde{t})}\leq \frac{1}{8n}
\end{equation*}
with $\tilde{t}=\log(|A|)/n$ in our setting and since $\kappa\,\log(512\kappa)\,\tilde{t}/\log(512/\tilde{t})\geq 1$. The latter holds since $\tilde{t}\geq \frac{1}{\kappa}$ by assumption. By the Main Lemma ($\tilde{t}_0=\frac{1}{\kappa}$ and $l\leq n- \frac{n}{\kappa} - 1$ since $\frac{n}{\kappa}\leq \log(|A|)-l-1$ by assumption)
\begin{equation*}
 (1-P)^n \geq \bigg(1-\frac{1}{8n}\bigg)^n\geq 7/8.
\end{equation*}
We estimate the probability that $\h$ is $k$-local next. Let $|V_i|$ denote the number of variables $\h_i$ depends on. By Chernoff's Bound
\begin{equation*}
 \Pr_{\h_i}(|V_i|\geq 2 p n)=\Pr(|V_i|\geq k+1)\leq (\mathrm{e}/4)^{p n}\leq \frac{1}{16n}.
\end{equation*}
Thus,
\begin{equation*}
 \Pr_\h(\forall i:\,|V_i| \leq k)\geq \bigg(1-\frac{1}{16n}\bigg)^n\geq 7/8.
\end{equation*}
The joint probability that Eq.\ \rf{eq:12} holds and $\h$ is $k$-local is thus at least $2\cdot\frac{7}{8}-1=\frac{3}{4}$. The inner loop amplifies this probability to $1-1/n$.\\

\noindent\textbf{Case} $|A| 2^{-l+3}<1$. Let $X=X(\h,b)$ be $|B|$. Let $X_x$ indicate whether $\h(x)=b$. Then, $\E(X_x)=2^{-l}$ and thus $\E(X)=|A| 2^{-l}$. By Markov's Inequality, $$\Pr_{\h,b}(X < 1)\geq 7/8.$$
This implies that the joint probability that $B=\{\}$ and $\h$ is $k$-local is at least $\frac{3}{4}$. The inner loop amplifies this probability to $1-1/n$.\\

\noindent Eq.\ \rf{eq:12} implies $B\not=\{\}$ if $A\not=\{\}$. Assume the algorithm stops at $l=l_0\in[n]$. It outputs $2^{l_0-1}$. From the first case, we get that $l_0 \geq \log(|A|) - \frac{n}{\kappa}-1$ with probability (w.p.) at least $(1-1/n)^{n+1}$ because the algorithm continues if $l_0 < \log(|A|) - \frac{n}{\kappa}-1$ w.p.\ at least $1-1/n$ per step. From the second case, $l_0 \leq \log(|A|)+3$ w.p.\ at least $1-1/n$ because the algorithm stops if $l_0 > \log(|A|) + 3$ w.p.\ at least $1-1/n$.\\

We do not know how the algorithm behaves in the range $\Omega(1)\leq\log(|A|)\leq \frac{n}{\kappa}+O(1)$. This causes the approximation error. We can overcome this problem using a simply technique to prove the second claim.\\

\noindent\textbf{Claim 2 (constant case).} This analysis remains the same as in the non-constant case. We are just have to show $(1-Q)^m\geq 1/4$ which follows from $Q\leq \frac{1}{m}$ and $\eps:=\frac{1}{2}$, $\zeta:=1-\frac{4}{k}$, $t_0:=n-\frac{\log(n)}{k}\,n^{1-4/k}$.\\


\noindent\textbf{Claim 1.} We note that we can count the number of solutions exactly in time $\OO(2^{(c+\delta)n})$ if $|A|\leq 2^{\delta n}$. This follows from the self-reducibility of SAT and the prerequisites. Set $\delta:=\frac{1}{\kappa}$. If $p\leq \frac{\delta}{2}$ we know that $\h$ is with high probability $(\delta n)$-local. We can encode a $(\delta n)$-local hash function in time $\OO(2^{\delta n})$ as a CNF. We adapt $\textsc{acount}$ in the following way: If the input CNF $F$ has more than $\lfloor 2^{\delta n}\rfloor$ solutions we construct $\h$ for $l=1,...,\lceil(1-\delta)n\rceil$ and continue as long as $F\wedge G$ has at least $\lfloor 2^{\delta n}\rfloor$ solutions. We output the exact number of solutions of $F\wedge G$ times $2^l$. The analysis goes as follows. We observe that as soon as $|A|\,2^{-l}<\lfloor 2^{\delta n}\rfloor$ we know it and the approximation error is thus determined by Eq.\ \ref{eq:12}. Rewriting Eq.\ \ref{eq:12} we get 
\begin{equation*}
(1-\eps)\,|A| \leq |B|\,2^{l}\leq (1+\eps)\,|A|.
\end{equation*}
For some $p=O(\delta)$ and $\tilde{t}\geq \delta$ we get from the Main Lemma
\begin{equation*}
 P = 2^{-O(\delta^2 n / \log(1/\delta)-\log(\frac{1}{\eps}))}
\end{equation*}
which is small enough for some $\log(\frac{1}{\eps})=\Omega(\delta^2 n / \log(1/\delta))$. 
\end{proof}

\end{document}